\documentclass{llncs}

\usepackage{latexsym}
\usepackage{amsmath}
\usepackage{amsfonts}
\usepackage{amssymb}
\usepackage{xspace}
\usepackage[utf8]{inputenc}
\usepackage{cite}
\usepackage[usenames]{color} %if using pdfLaTeX.
\usepackage[noend,ruled,linesnumbered]{algorithm2e}
\usepackage{url}
\usepackage{subfigure}
\usepackage{ifpdf}
\ifpdf
	\usepackage[pdftex]{graphicx}
\else	
	\usepackage{graphicx}
\fi
\DeclareGraphicsExtensions{.pdf}

\newcommand{\parhead}[1]{{\smallskip\noindent\textbf{#1.}\xspace}}

\newcommand{\myboldmath}{}%\boldmath doesn't work with latex
\newcommand{\defn}[1]           {{\textit{\textbf{\myboldmath #1}}}}

\newcommand{\RNs}{{Radio Networks}\xspace}

\newcommand{\DNs}{{Dynamic Networks}\xspace}

\newcommand{\SNs}{{Sensor Networks}\xspace}
\newcommand{\FU}{{FU}\xspace}
\newcommand{\MD}{{MD}\xspace}
\newcommand{\MDFU}{{MDFU}\xspace}
\newcommand{\MDFULP}{{MDFU-LP}\xspace}
\newcommand{\FUname}{{Flow-Updating}\xspace}
\newcommand{\MDname}{{Mass-Distribution}\xspace}

\newcommand{\mig}[1]{\textcolor{blue}{#1}\marginpar{MM}}
\newcommand{\pj}[1]{\textcolor{red}{#1}\marginpar{PJ}}

\renewcommand{\mig}[1]{#1}
\renewcommand{\pj}[1]{#1}

\begin{document}

\mainmatter              % start of the contributions
%---------------------------------------------------------

\title{
Fault-Tolerant Aggregation:\\
\FUname Meets \MDname
\thanks{This work is supported in part by the Comunidad de Madrid grant
S2009TIC-1692, Spanish MICINN grant TIN2008--06735-C02-01, Portuguese FCT
grant PTDC/EIA-EIA/104022/2008, and National Science Foundation grant CCF-0937829.}}
\titlerunning{Fault-Tolerant Aggregation}  % abbreviated title (for running head)
%                                     also used for the TOC unless
%                                     \toctitle is used
%
\author{
   Paulo~Sérgio~Almeida\inst{1}
      \and
   Carlos~Baquero\inst{1}
	\and
   Martín~Farach-Colton\inst{2}
      \and \\
   Paulo~Jesus\inst{1}  
   \and
   Miguel~A.~Mosteiro\inst{3}
}
\authorrunning{P. S. Almeida, C. Baquero, M. Farach-Colton, P. Jesus, and M. A. Mosteiro}   % abbreviated author list (for running head)
\institute{
Depto. de Inform\'atica (CCTC-DI), Universidade do Minho, Braga, Portugal\\
\email{\{psa,cbm,pcoj\}@di.uminho.pt}
   \and
Dept. of Computer Science, Rutgers University, Piscataway, NJ, USA \&\\ 
Tokutek, Inc.\\ \email{farach@cs.rutgers.edu}
   \and 
Dept. of Computer Science, Rutgers University, Piscataway, NJ, USA \&\\ 
LADyR, GSyC, Universidad Rey Juan Carlos, Madrid, Spain\\
\email{mosteiro@cs.rutgers.edu}
}

\maketitle

%%%%%%%%%%%%%%%%%%%%%%%%%%%%%%%%%%%%%%%%%%%%%%%%%%%%%%

\begin{abstract}
  \FUname (FU) is a fault-tolerant technique that has proved to be
  efficient in practice for the distributed computation of aggregate
  functions in communication networks where individual processors do
  not have access to global information. Previous distributed
  aggregation protocols, based on repeated sharing of input values (or
  \defn{mass}) among processors, sometimes called \MDname
  (MD) protocols, are not resilient to communication failures (or
  \defn{message loss}) because such failures yield a loss of mass.

  In this paper, we present a protocol which we call
  \defn{\MDname with \FUname (MDFU)}.  We obtain MDFU by
  applying FU techniques to classic MD.  We analyze the convergence
  time of MDFU showing that stochastic message loss produces low
  overhead. This is the first convergence proof of an FU-based
  algorithm.  We evaluate MDFU experimentally, comparing it with
  previous MD and FU protocols, and verifying the behavior predicted
  by the analysis.  Finally, given that MDFU incurs a fixed deviation
  proportional to the message-loss rate, we adjust the accuracy of
  MDFU heuristically in a new protocol called \defn{MDFU with Linear
  Prediction (MDFU-LP)}.  The evaluation shows that both MDFU and
  MDFU-LP behave very well in practice, even under high rates of
  message loss and even changing the input values dynamically.

\parhead{Keywords} Aggregate computation, Distributed computing, Radio networks, Communication networks.
\end{abstract}

\newcommand{\footnotenonumber}[1]{{\def\thempfn{}\footnotetext{#1}}}
\footnotenonumber{The authors appear in alphabetical order.}

%------------------------------------------------------------------------- 

\section{Introduction}

The distributed computation of algebraic aggregate functions is particularly challenging in settings where the processing nodes do not have access to global information such as the input size. A good example of such scenario is \SNs~\cite{akyildiz-survey,ravi-survey} where unreliable sensor nodes are deployed at random and the overall number of nodes that actually start up and sense input values may be unknown.
\mig{Under such conditions, well-known techniques for distributing
information throughout the network such as Broadcast~\cite{KP:adaptvsobliv}
or Gossiping~\cite{G:gossiping} cannot be directly applied, }\pj{and data collection is only practicable if aggregation is performed.}
Even more challenging is that loss of messages between nodes or even node crashes are likely in such harsh settings. 
It has been proved~\cite{BGMGM:agg} that the problem of aggregating values distributedly in networks where processing nodes may join and leave arbitrarily is intractable.
Hence, arbitrary adversarial message loss also yields the problem intractable, but a weaker adversary, for instance a stochastic one as in \DNs~\cite{CPMS:dynRN}, is of interest.
In this paper, under a stochastic model of message loss, we study communication networks where each node holds an input value and the average of those values~\footnote{Other algebraic aggregate functions can be computed in the same bounds using an average protocol~\cite{CPX:SNaggJournal, KDGgossip}.} must be obtained by all nodes, none of whom have access to global information of the network, \mig{\emph{not even the total number of nodes $n$}.}

A classic distributed technique for aggregation, sometimes called \defn{\MDname}~(\MD)~\cite{FMT:agg}, works in rounds. In each round, each node shares a fraction of its current average estimation with other nodes, starting from the input values~\cite{KDGgossip, BGPSgossip, Chen09, CPX:SNaggJournal,XBL:distSensorFusion, XB:distAvgCons, OSM:ConsProb, SP:DistAvgCons}. Details differ from paper to paper but a common problem is that, in the face of message loss, those protocols \mig{ either do not converge to a correct output or they require some instantaneous failure detector mechanism that updates the topology information at each node in each round.}
Recently~\cite{ABJ:flowupdate, ABJ:flowupdatedyn}, a heuristic termed \defn{\FUname} (\FU) addressed the problem assuming stochastic message loss~\cite{ABJ:flowupdate}, and even assuming that input values change and nodes may fail~\cite{ABJ:flowupdatedyn}.
The idea underlying \FU is to keep track of an aggregate function of all communication for each pair of communicating nodes, since the beginning of the protocol, so that a current value at a node can be re-computed from scratch in each round.
Empirical evaluation has shown that \FU behaves very well in practice~\cite{ABJ:flowupdate, ABJ:flowupdatedyn}, but  such protocols have eluded analysis until now.

%\parhead{Our Contributions}
In this paper, we introduce the concept of \FU to \MD. First, we present a protocol that we call \defn{\MDname with \FUname} (\MDFU). The main difference with \MD is that, instead of computing incrementally, the average is computed from scratch in each round using the initial input value and the accumulated value shared with other nodes so far (which we refer to as either \defn{mass shared}, or \defn{flow passed}). The main difference with \FU is that if messages are not lost the algorithm is exactly \MD, which facilitates the theoretical analysis of the convergence time under failures parameterized by the failure probability (or \defn{message-loss rate}).

\parhead{Our results}%The analysis shows the following.
We first leverage previous work on bounding the mixing time of Markov chains~\cite{SJ:count} to show that, for any $0<\xi<1$, the convergence time of \MDFU under reliable communication is 
$2\ln(n/\xi)/\Phi(G)^2$, 
where $\Phi(G)$ is the conductance of the underlying graph characterizing the execution of \MDFU on the network.
Then, we show that, with probability at least $1-1/n$,
for a message-loss rate $f < 1/\ln(2\Delta e)^3$, the multiplicative overhead on the convergence time produced by message loss is less than $1/(1-\sqrt{f\ln(2\Delta e)^3})$, and it is constant for $f\leq 1/(e(2\Delta e)^e)$, where $\Delta$ is the maximum number of neighbors of any node.
Also, we show that, with probability at least $1-1/n$, 
for any $0<\xi<1$, after convergence 
the expected average estimation at any node is in the interval $[(1-\xi)(1-f) \overline{v},(1+\xi)\overline{v}]$.
This is the first convergence proof for an \FU-based algorithm.

In \MDFU, if some flow is not received, a node computes the current estimation using the last flow received. Thus, in presence of message loss, nodes do not converge to the average and only some parametric bound can be guaranteed as shown. Aiming to improve the accuracy of \MDFU, we present a new heuristic protocol that we call \defn{MDFU with Linear Prediction} (MDFU-LP). The difference with \MDFU is that if some flow is not received a node computes the current estimation using an estimation of the flow that should have been received.

We evaluate \MDFU and \MDFULP experimentally and find that the performance
of \MDFU is comparable to \FU and other competing algorithms under reliable
communication. In the presence of message loss, the empirical evaluation
shows that \MDFU behaves as predicted in the analysis converging to the
average with a bias proportional to the message-loss rate.
This bias is not present in the original \FU, which converges to the
correct value even under message loss.
%On the other hand, we observe that \MDFU converges smoothly to the average,
%with root mean square error decreasing monotonically, without exhibiting the
%erratic behavior in the initial rounds as the original \FU.
In a third set of evaluations, we observe that \MDFULP converges to the
correct value even under high message loss rates, with the
same speed as under reliable communication.
We also test \MDFU under changing input values to verify that it tolerates
dynamic changes in practice, in contrast to classic \MD algorithms, which
need to restart the computation each time values are changed. 

\parhead{Roadmap}
In Section~\ref{sec:prelim} we formally define the model and the problem,
and we give an overview of related work. Section~\ref{sec:alg} includes the
details of \MDFU and its analysis, whereas its empirical evaluation is
covered in Section~\ref{sec:eval}. In Section~\ref{sec:vel} we present the
details of \MDFULP and its experimental evaluation. Section
\ref{sec:value-change} evaluates \MDFU
in a dynamic setting, where input values change over time.

\section{Preliminaries}
\label{sec:prelim}
\parhead{Model}
We consider a static connected communication network formed by a set $V$ of $n$ processing \defn{nodes}.
We assume that each node has an identifier (ID).
Any pair of nodes $i,j \in V$ such that $i$ may send messages to $j$ without relying on other nodes (one hop) are called \defn{neighbors}.
We assume that the IDs are assigned so that each node is able to distinguish all its neighbors.
The set of ordered pairs of neighbors (or, \defn{edges}) is called $E$.
The network is symmetric, meaning that, for any $i,j\in V$, $(i,j)\in E$ if and only if $(j,i)\in E$.
The set of neighbors of a given node $i$ is denoted as $N_i$ and $|N_i|$ is called the \defn{degree} of $i$.
For each pair of nodes $i,j \in V$, the maximum degree between $i,j$ is denoted as $D_{ij}=\max\{|N_i|,|N_j|\}$.
The maximum degree throughout the network is denoted as $\Delta=\max_{i\in V}|N_i|$.
Each node $i$ knows $N_i$ and $D_{ij}$ for each $j\in N_i$, but does not know the size of the whole network $n$.
The time is slotted in \defn{rounds} and each round is divided in two \defn{phases}. 
In each round, a node is able to send (resp. receive) one message to (resp. from) all its neighbors (communication phase) and to perform local computations (computation phase).
However, for each $(i,j)\in E$ and for each communication phase, a message from $i$ to $j$ is lost independently with probability $f$. 
\mig{This is a crucial difference with previous work where, although edge-failures are considered, messages are not lost thanks to the availability of some failure detection mechanism. More details are given in the previous work section.}
Nodes are assumed to be reliable, i.e. they do not fail.

\parhead{Problem}
Each node $i$ holds an input value $v_i$, for $1\leq i\leq n$. 
The aim is for each node to compute the average $\overline{v}=\sum_{i=1}^n v_i/n$ without any global knowledge of the network. 
We focus on the algorithmic cost of such computation, counting only the number of rounds that the computation takes after simultaneous startup of all nodes, leaving aside medium access issues to other layers. 
This assumption could be removed as in~\cite{FMT:agg}.

\parhead{Previous Work}
Previous work on aggregate computations has been particularly prolific for the area of \RNs, including both theoretical and experimental work~\cite{KBCHLrobust, DSWgeoGossJournal, DBLP:conf/dsn/GuptaRB01, JMBagg, KDGgossip, ZGEagg, MFHHtag, NGSAagg, MSFCagg, DBLP:conf/icdcsw/KrishnamachariEW02, intanagonwiwat00directedJournal, DBLP:conf/icdcs/IntanagonwiwatEGH02, DBLP:conf/sosp/HeidemannSIGEG01}. Many of those and other aggregation techniques exploit global information of the network~\cite{FMT:agg, DBLP:conf/icdcsw/KrishnamachariEW02, MFHHtag, DBLP:conf/dsn/GuptaRB01}, or are not resilient to message loss~\cite{BGPSgossip, KDGgossip, Chen09}.

\FU is a recent fault-tolerant approach\cite{ABJ:flowupdate,ABJ:flowupdatedyn} inspired on the concept of flows (from graph theory). Like common \MD techniques, it is based on the execution of an iterative averaging process at all nodes, and all estimates eventually converge to the system-wide average.
\MD protocols exchange ``mass'', which lead them to converge to a wrong result in the case of message loss. In contrast, \FU does not exchange ``mass''.  Instead it performs idempotent flow exchanges which provide resilience against message loss. In particular, \FU keeps the initial input value at each node unchanged (in a sense, always conserving the global mass), exchanging and updating flows between neighbors for them to produce a new estimate. The estimate is computed at each node from the input values and the contribution of the flows. 
No theoretical bounds on the performance of the algorithm were provided. 
Empirical evaluation shows that \FU performs better than classic \MD algorithms, especially in low-degree networks, and it supports high levels of message loss~\cite{ABJ:flowupdate}. Moreover, it self-adapts to dynamic changes (i.e. nodes leaving/arriving and input value change) without any restart mechanism (like other approaches), and tolerates node crashes~\cite{ABJ:flowupdatedyn}.

\MD protocols for average computations in arbitrary networks based on gossiping (exchange values in pairs) were studied in~\cite{BGPSgossip,KDGgossip}. 
Results in~\cite{BGPSgossip} are presented for all gossip-based algorithms by characterizing them by a matrix that models how the algorithm evolves while sharing values in pairs iteratively. As in our results, the time bounds shown are given as a function of the spectral decomposition of the graph underlying the computation.
\mig{The work is focused on optimizing distributedly the spectral gap, in order to minimize convergence time. The dynamics of the model are motivated by changes in topology induced by nodes leaving and joining the network. Those changes may be introduced in the probability of establishing communication between any two nodes. However, the delivery of messages has to be reliable to ensure mass conservation.}
An algorithm called Push-Sum that takes advantage of the broadcast nature of \RNs (i.e., it is not restricted to gossip) is included in~\cite{KDGgossip}, yielding similar bounds.
Chen, Pandurangan, and Hu~\cite{Chen09} present an \MD algorithm that first builds a forest over the network, where each root collects the information, and then a gossiping algorithm among the roots is used. The authors show a reduction on the energy consumption with respect to the uniform gossip algorithm. 
On the other hand, the \MD algorithm presented in~\cite{CPX:SNaggJournal} relies on a different randomly chosen local leader in each round to distribute values. The bounds given are also parameterized by the eigen-structure of the underlying graph. This result was extended more recently~\cite{CH:aggSHS} to networks with a time-varying connection graph, \mig{but the protocol requires to update the matrix underlying such graph in each round.}

\mig{\MD protocols have been used also for Distributed Average Consensus~\cite{XBL:distSensorFusion, Xiao:2006jn, XB:distAvgCons, OSM:ConsProb, SP:DistAvgCons, Spanos:2005ux} within Control Theory, but they do not apply to our model. 
For example, in~\cite{XBL:distSensorFusion, Xiao:2006jn} the model includes unreliable communication links, but the algorithm requires instantaneous update of the topology information held at each node at the beginning of each round. Others, either rely on similar features~\cite{OSM:ConsProb, SP:DistAvgCons, Spanos:2005ux} or do not consider changes in topology at all~\cite{XB:distAvgCons}.
}

The common problem in all the \MD protocols is that they are not resilient to message loss, because it implies a loss of mass. Hence, if messages are lost, they need to restart the computation from scratch. In MDFU, message loss has an impact on convergence time, which  we show to be small, but the computation recovers from those losses, yielding the correct value. In fact, it is this characteristic of MDFU and FU in general what makes the technique suitable for dynamic settings in which the input values change with time.

\section{\MDFU}
\label{sec:alg}
\begin{algorithm}[htbp]
\label{alg}
\caption{\MDFU. Pseudocode for node $i$. $e_i$ is the estimate of node $i$. $F_{in}(j)$ is the cumulative inflow from node $j$. $F_{out}(j)$ is the cumulative outflow to node $j$. 
}
\dontprintsemicolon

\tcp{initialization}
$e_i \leftarrow v_i$\;
\ForEach{$j\in N_i$}{
$F_{in}(j) \leftarrow 0$\;
$F_{out}(j) \leftarrow e_i/\left(2D_{ij}\right)$\;
}
\BlankLine\;

\ForEach{round}{
\tcp{communication phase}
\ForEach{$j\in N_i$}{
Send $j$ message $\langle i,F_{out}(j)\rangle$\;
}
\ForEach{$\langle j,F\rangle$ received}{
$F_{in}(j) \leftarrow F$\;
}
\tcp{computation phase}
$e_i \leftarrow v_i + \sum_{j\in N_i} (F_{in}(j) - F_{out}(j)) $\;
\ForEach{$j\in N_i$}{
$F_{out}(j) \leftarrow F_{out}(j) + e_i/\left(2D_{ij}\right)$\;
}
}
\end{algorithm}

As in previous work~\cite{KDGgossip, CPX:SNaggJournal, BGPSgossip, FMT:agg}, \MDFU is based on repeatedly sharing among neighbors a fraction of the average estimated so far.
Unlike in those papers, in \MDFU the estimation is computed from scratch in each round, as in \FU~\cite{ABJ:flowupdate,ABJ:flowupdatedyn}. 
For that purpose, each node keeps track of the cumulative value passed to each neighbor (or, cumulative flow) since the protocol started. 
Together with the original input value, those flows allow each node to recompute the average estimation in each round.
Should some flow from node $i$ to node $j$ be lost, $j$ temporarily computes the estimation using the last flow received from $i$.
Further details can be found in Algorithm~\ref{alg}.

%\section{Analysis of \MDFU}
\label{sec:ana}
Recall that the aim is to compute the average $\overline{v}=\sum_{i=1}^n v_i/n$ of all input values.
Let $e_i(r)$ be the average \defn{estimate} of node $i$ in round $r$, and $\varepsilon(r)=\max_i \{|e_i(r)-\overline{v}|/\overline{v}\}$ be the maximum relative \defn{error} of the average estimates in round $r$.
We want to bound the number of rounds after which the maximum relative error is below some parametric value $\xi$.

In each round, a node shares a fraction of its current estimate with each neighbor. Therefore, the execution of each round can be characterized by a \defn{transition matrix}, denoted as $\mathbf{P} = (p_{ij})$, $\forall i,j \in V$, such that for any round $r$ where messages are not lost
\begin{align*}
p_{ij} = \left\{ \begin{array}{ll} 
1/(2D_{ij}) & \textrm{if $i\neq j$ and $(i,j) \in E$,}\\
1-\sum_{k\in N_i}1/(2D_{ik}) & \textrm{if $i=j$,}\\
0 & \textrm{$(i,j) \notin E$}
\end{array} \right. 
\end{align*}
and $\mathbf{e}(r+1)=\mathbf{e}(r) \mathbf{P}$, where $\mathbf{e}(\cdot)$ is the row vector $(e_1(\cdot) e_2(\cdot) \dots e_n(\cdot))$.

%%%%%%%%%%%%%%%%%%%%%%%%%%%%%%%%%%%%%%%%%%%%%%%%%%%%%

\subsection{Convergence Time for $f=0$}
\label{sec:nofail}
Consider first the case when the communication is reliable, that is $f=0$. Then, the above characterization is round independent and, given that $\mathbf{P}$ is stochastic, it can be seen as the transition matrix of a time-homogeneous Markov chain $(X_r)_{r=1}^\infty$ 
with finite state space $V$. Furthermore, $(X_r)_{r=1}^\infty$ is irreducible, and aperiodic, then it is ergodic and it has a unique stationary distribution. Given that $\mathbf{P}$ is doubly stochastic such stationary distribution is $\pi_i=1/n$ for all $i\in V$. 
Thus, bounding the convergence time of $(X_r)_{r=1}^\infty$ we have a bound for the convergence time of \MDFU without message loss.
The following notation will be useful.
Let $G$ be a weighted undirected graph with set of nodes $V$ and where, for each pair $i,j \in V$, the edge $(i,j)$ has weight $\pi_i p_{ij}$. $G$ is called the underlying graph of the Markov chain $(X_r)_{r=1}^\infty$. 
The following quantity characterizes the likelihood that the chain does not stay in a subset of the state space with small stationary probability.
Let the \defn{conductance} of graph $G$ be
$$\Phi(G)=\min_{\substack{\emptyset\subset S\subset V\\\sum_{i\in S} \pi_i\leq1/2}} \frac{\sum_{i,j\in S} p_{ij}\pi_i}{\sum_{i\in S} \pi_i}.$$
The following theorem shows the convergence time of \MDFU with reliable communication parameterized in the conductance of $G$.

\begin{theorem}
\label{thm:succ}
For any communication network of $n$ nodes running \MDFU, 
for any $0<\xi<1$, 
and for $r_c=2\ln(n/\xi)/\Phi(G)^2$, 
if $f=0$, it holds that $\varepsilon(r)\leq \xi$ for any round $r\geq r_c$,
where $\Phi(G)$ is the conductance of the underlying graph characterizing the execution of \MDFU on the network.
\end{theorem}

\begin{proof}
We want to find a value of $r_c$ such that for all $r\geq r_c$ it holds that $\max_i \{|e_i(r)-\overline{v}|/\overline{v}\}\leq \xi$. Then, we want $\max_i \{|e_i(r)/\sum_{j\in V}v_j-1/n|\}\leq \xi / n$. Given that $e_i(r)=\sum_{j\in V}v_j(\mathbf{P}^r)_{ji}$, it is enough to have $\max_{j,i \in V}\{|(\mathbf{P}^r)_{ji}-1/n|\}\leq \xi / n$.
On the other hand, given that $p_{ij}\pi_i=p_{ji}\pi_j$ for all $i,j \in V$, the Markov chain %$(X_r)_{r=1}^\infty$ 
is time-reversible. Then, as proved in~\cite{SJ:count}, it is $\max_{i,j\in V}|(\mathbf{P}^r)_{ij}-\pi_j|/\pi_j \leq \lambda^r_1 / \min_{j\in V} \pi_j$, where $\lambda_1$ is the second largest eigenvalue of $\mathbf{P}$ (all the eigenvalues of $\mathbf{P}$ are positive because $p_{ii}\geq 1/2$ for all $i\in V$). Given that $\pi_i=1/n$ for all $i\in V$, we have $\max_{i,j\in V}|(\mathbf{P}^r)_{ij}-1/n| \leq \lambda^r_1$. Thus, from the inequality above, it is enough to have $\lambda^r_1\leq \xi / n$. 
As proved also in~\cite{SJ:count}, given that $(X_r)_{r=1}^\infty$ is ergodic and time-reversible, it is $\lambda_1\leq 1-\Phi(G)^2/2$.
Then, it is enough $(1-\Phi(G)^2/2)^r\leq \xi / n$.
Given that $\Phi(G)\leq 1$, using that $1-x\leq e^{-x}$ for $x<1$, the claim follows. 
%\begin{align*}
%\exp\left(-\frac{r\Phi(G)^2}{2}\right) &\leq \frac{\xi}{n}\\
%\frac{r\Phi(G)^2}{2} &\geq \ln\frac{n}{\xi}\\
%r &\geq \frac{2 \ln(n/\xi)}{\Phi(G)^2}\\
%\end{align*}
\end{proof}

%%%%%%%%%%%%%%%%%%%%%%%%%%%%%%%%%%%%%%%%%%%%%%%%%%%%%%

\subsection{Convergence Time for $f>0$}
\label{sec:fail}
\label{sec:rw}

\parhead{Mixing time of a multiple random walk}
Recall that we carry out an average computation of $n$ input values where each node $i$ shares a $1/(2D_{ij})$ fraction of its estimate in each round of the computation with each neighboring node $j$. We have characterized each round of the computation with a transition matrix $\mathbf{P}$ so that in each round $r$ the vector of estimates $\mathbf{e}(r)$ is multiplied by $\mathbf{P}$.

The Markov chain defined in Section~\ref{sec:nofail} that models the average computation is also a characterization of a random walk, that is, a stochastic process on the set of nodes $V$ where a particle moves around the network randomly. In our case, for each round, instead of choosing the next node where the particle will be located uniformly among neighbors, the matrix of transition probabilities is $\mathbf{P}$. A state of this process (which of course is also Markovian) is a distribution of the location of the particle over the nodes. The measure of this random walk that becomes relevant in our application is the mixing time, that is, the number of rounds before such distribution will be \emph{close} to uniform. The mixing time of this random walk is the same as the convergence time of the Markov chain $(X_r)_{r=1}^\infty$, setting appropriately for each case the desired maximum deviation with respect to the stationary distribution as follows. 

A useful representation of this process in our application is to assume a set $S$ of particles, all of the same value $\nu$, so that at the beginning each node $i$ holds a subset $S_i$ of particles such that $|S_i|\nu=v_i$. In order to analyze the computation along many rounds, we assume that $\nu$ is small enough so that particles are not divided. 
We define the \defn{mixing time} of this multiple random walk as the number of rounds before the distribution of all particles is within $\xi/n$ of the uniform, for $0<\xi<1$.
Without message loss, it can be seen that the mixing time of the above defined multiple random walk is the same as the convergence time of the Markov chain $(X_r)_{r=1}^\infty$ defined in Section~\ref{sec:nofail}. We consider now the case where messages may be lost.

%------------------------------

The following lemma shows that, for $f < 1/\ln(2\Delta e)^3$, the multiplicative overhead on the mixing time produced by message loss is less than $1/(1-\sqrt{f\ln(2\Delta e)^3})$, and it is constant for $f\leq 1/(e(2\Delta e)^e)$.
The proof 
%, left to the full version of this work in~\cite{MDFUarXiv} for brevity, 
uses concentration bounds on the delay that any particle may suffer due to message loss.

\begin{lemma}
\label{lemma:rw}
Consider any communication network of $n$ nodes running \MDFU, 
any $0 < f \leq 1/\ln(2\Delta e)^3$,
any $0<\xi<1$, 
let $r_c = 2\ln(n/\xi)/\Phi(G)^2$,
and let
%\begin{align*}
%r = \frac{r_c}{1-1/e} &\textrm{ if } f\leq \frac{1}{e(2\Delta e)^e} \textrm{ or}\\
%r = \frac{r_c}{1-f\left(\sqrt{4\ln(2\Delta e)^3/f-3}-1\right)/2}  &\textrm{ if } f< \frac{1}{\ln(2\Delta e)^3-1}.
%\end{align*}
\begin{align*}
q = \left\{ \begin{array}{ll} 
1/e & \textrm{ if $f\leq 1/(e(2\Delta e)^e)$}\\
f \left(\sqrt{4\ln(2\Delta e)^3/f-3}-1\right)/2  & \textrm{ otherwise.}
\end{array} \right.
\end{align*}
Consider a multiple random walk modeling \MDFU as described.
With probability at least $1-1/n$, after $r=r_c/(1-q)$ rounds it holds that
$\max_{x\in S,i\in V}|p_x(i)-1/n|\leq \xi/n$, where $p_x(i)$ is the probability that particle $x$ is located at node $i$.
\end{lemma}

\begin{proof}
For clarity, we model the network with a directed graph $\{V,E\}$, with $V$ and $E$ as defined in the model.
A message loss in the edge $(i,j)\in E$ is modeled with a buffer on the edge $(i,j)$ where a particle is ``delayed''. 
For a computation of $r$ rounds, it is enough to consider at most $n(2\Delta)^r$ particles, because initially there are $n$ input values and each value is divided $r$ times by at most $2\Delta$.
Consider the random walk of a given particle $x\in S$.
For each round, $x$ is delayed with probability $f$. 
We bound the mixing time by bounding the number of rounds that any particle is delayed as follows.

Assume first that $1/(e(2\Delta e)^e) < f \leq 1/\ln(2\Delta e)^3$.
For $r$ rounds, the expected number of rounds when a given particle is delayed is $fr$.
%\begin{align*}
%\mathbf{E}(\textrm{delay out of $r$ rounds})=fr.
%\end{align*}
Using Chernoff-Hoeffding bounds~\cite{book:mitzenmacher},
the probability that a given particle $x$ is delayed more than $qr$ rounds, $f\leq q\leq 1$, is at most $ \exp(-fr(q/f-1)^2/3)$.
%\begin{align*}
%\mathbf{Pr}(\textrm{$x$ delayed $>qr$ rounds})\leq \exp(-fr(q/f-1)^2/3).
%\end{align*}
Then, the probability that \emph{some} particle is delayed more than $qr$ rounds is 
%at most $n\left(\frac{2\Delta}{\exp((q-f)^2/(3f))}\right)^r$.
\begin{align*}
\mathbf{Pr}(\textrm{$\exists x: x$ delayed $>qr$}) &\leq n\left(\frac{2\Delta}{\exp((q-f)^2/(3f))}\right)^r.
%\label{overallprob}
\end{align*}
Assuming that $2\Delta \exp(1-q) \leq \exp((q-f)^2/(3f))$, we get that 
\begin{align*}
\mathbf{Pr}(\textrm{$\exists x: x$ delayed $>qr$}) &\leq n\left(\frac{1}{\exp(1-q)}\right)^{\frac{2\ln(n/\xi)}{(1-q)\Phi(G)^2}}\\
&= n \exp\left(-\frac{2\ln(n/\xi)}{\Phi(G)^2}\right), \textrm{given that $\xi\leq 1$ and $\Phi(G)\leq 1$,}\\
&\leq n \exp(-2\ln n)\\
&= 1/n.
\end{align*}
Then, it remains to prove
\begin{align*}
%\frac{2\Delta}{\exp((q-f)^2/(3f))} &\leq \frac{1}{\exp(1-q)}\\
2\Delta \exp(1-q) &\leq \exp((q-f)^2/(3f))\\
%\ln(2\Delta)+1-q &\leq (q-f)^2/(3f)\\
%3f(\ln(2\Delta)+1-q) &\leq (q-f)^2\\
%f\ln(2\Delta e)^3 &\leq (q-f)^2+3fq\\
%f\ln(2\Delta e)^3 &\leq q^2+fq+f^2\\
q^2+fq+f^2-f\ln(2\Delta e)^3 &\geq 0.
%q &\gtrless \left(-f\pm\sqrt{f^2-4(f^2-f\ln(2\Delta e)^3)}\right)/2\\
%q &\geq f\left(\sqrt{4\ln(2\Delta e)^3/f-3}-1\right)/2.
\end{align*}
Which is true for $q = f\left(\sqrt{4\ln(2\Delta e)^3/f-3}-1\right)/2$, which is feasible because, for $f\leq 1/\ln(2\Delta e)^3$, such value of $q$ implies $f\leq q\leq 1$ .

Consider now the case $0 < f \leq 1/(e(2\Delta e)^e)$.
Again, using Chernoff-Hoeffding bounds, the probability that a given particle $x$ is delayed more than $qr$ rounds, $f\leq q\leq 1$, is at most $\left((fe/q)^q/e^{f}\right)^{r}$
%\begin{align}
%\mathbf{Pr}(\textrm{$x$ delayed $>qr$ out of $r$ rounds}) &\leq \left(\frac{e^{(q/f-1)}}{(q/f)^{q/f}}\right)^{fr}\nonumber\\
%&= \left(\frac{e^{q-f}}{(q/f)^{q}}\right)^{r}\nonumber\\
%&= \left(\frac{1}{e^{f}}\left(\frac{fe}{q}\right)^q\right)^{r}.\label{forone}
%\end{align}
Then, the probability that \emph{some} particle is delayed more than $qr$ rounds is
\begin{align*}
\mathbf{Pr}(\textrm{$\exists x: x$ delayed $>qr$}) &\leq n\left(\frac{2\Delta}{e^{f}}\left(\frac{fe}{q}\right)^q\right)^{r}.
\end{align*}
Assuming that $2\Delta (fe/q)^q / e^{f} \leq 1/\exp(1-q)$ we get as before,
\begin{align*}
\mathbf{Pr}(\textrm{$\exists x: x$ delayed $>qr$}) &\leq 1/n.
\end{align*}
Then, it remains to prove
\begin{align*}
\frac{2\Delta}{e^{f}}\left(\frac{fe}{q}\right)^q &\leq \frac{1}{e^{1-q}}\\
%2\Delta\left(\frac{f}{q}\right)^q &\leq \frac{1}{ e^{1-f}}\\
2\Delta e^{1-f} &\leq \left(q/f\right)^q\\
2\Delta e &\leq \left(q/f\right)^q.
\end{align*}
Which is true for $f\leq 1/(e(2\Delta e)^e)$ and $q=1/e$.

\end{proof}

%%%%%%%%%%%%%%%%%%%%%%%%%%%%%%%%%%%%%%%%%%%%%%%%%%%

\parhead{The expected number of particles at each node as a function of $f$}
Analyzing a multiple random walk of a set of particles, in Lemma~\ref{lemma:rw} we obtained a bound on the time that any particle takes to converge to a stationary uniform distribution. However, for any probability of message loss $f>0$ and for any round, there is a positive probability that some particles are located in the edge buffers defined in the proof of such lemma. Hence, the fact that each particle is uniformly distributed over nodes does not imply that the expected average held at the nodes has converged, because only particles located at nodes are uniformly distributed. We bound the expected error in this section. The proof of the following lemma 
%, left to the full version of this work in~\cite{MDFUarXiv} for brevity, 
is based on computing the overall expected ratio of particles in nodes with respect to delayed particles.

\begin{lemma}
\label{lemma:error}
Consider a multiple random walk modeling \MDFU under the conditions of Lemma~\ref{lemma:rw}.
Then, with probability at least $1-1/n$, for any round $r\geq r_c/(1-q)$, the expected number of particles $\mathbf{E}(|S_i^{(r)}|)$ in each node $i$ is $(1-\xi)(1-f) |S|/n \leq \mathbf{E}(|S_i^{(r)}|) \leq (1+\xi)|S|/n$.
\end{lemma}

\begin{proof}
We consider a multiple random walk of a set of particles $S$ over a directed graph $V,E$, with $V$ and $E$ as defined in the model.
A message loss in the edge $(i,j)\in E$ is modeled with a buffer on the edge $(i,j)$ where a particle is ``delayed''. 
The following notation will be useful.
For any round $r$,
$S_X^{(r)}$ is the set of particles held at the set $X$ (node set or edge-buffer set),
%$S_{ij}^{(r)}$ is the set of particles held at the buffer of the directed edge $ij$,
and $S_i^{(r)}$ is the set of particles held at the node $i$.
%Let $E_i$ be the set of edge-buffers outgoing from node $i$.
Let $p_i = \sum_{j_\in N_i}  1/(2D_{ij})$ for any node $i$. 
By linearity of expectation, at the end of round $r$, the expected number of particles in all buffer-edges and the expected number of particles in all nodes are

%For any node $i\in V$ and round $r>0$, the expected number of particles at $E_i$ is, by linearity of expectation,
%\begin{align}
%\mathbf{E}(|S_{E_i}^{(r)}|) 
%&= \sum_{j_\in N_i} \mathbf{E}(|S_i^{(r-1)}|)  \frac{f}{2D_{ij}} + \sum_{j\in N_i} \mathbf{E}(|S_{ij}^{(r-1)}|) f\nonumber\\
%%&= \mathbf{E}(|S_i^{(r-1)}|) f \sum_{j_\in N_i}  \frac{1}{2D_{ij}} + f \sum_{j\in N_i} \mathbf{E}(|S_{ij}^{(r-1)}|) \nonumber\\
%%&= \mathbf{E}(|S_i^{(r-1)}|) f \sum_{j_\in N_i}  \frac{1}{2D_{ij}} + f \mathbf{E}(|S_{E_i}^{(r-1)}|) \nonumber\\
%&= \mathbf{E}(|S_i^{(r-1)}|) f p_i + f \mathbf{E}(|S_{E_i}^{(r-1)}|).\label{expedge}
%\end{align}

%And the expected number of particles at $i$ is, by linearity of expectation,
%\begin{align}
%\mathbf{E}(|S_i^{(r)}|) 
%&= \mathbf{E}(|S_i^{(r-1)}|) \left( 1- \sum_{j_\in N_i} \frac{f}{2D_{ij}} \right) + \sum_{j\in N_i} \mathbf{E}(|S_{ji}^{(r-1)}|) (1-f)\nonumber\\
%&= \mathbf{E}(|S_i^{(r-1)}|) ( 1- f p_i ) + (1-f) \sum_{j\in N_i} \mathbf{E}(|S_{ji}^{(r-1)}|).\label{expnode}
%\end{align}

%Summing Equation~\ref{expedge} over all nodes, by linearity of expectation,
\begin{align}
\mathbf{E}(|S_E^{(r)}|)
%&= \sum_{i\in V} \mathbf{E}(|S_{E_i}^{(r)}|)\nonumber\\ 
%&= \sum_{i\in V} \mathbf{E}(|S_i^{(r-1)}|) f p_i + \sum_{i\in V} f \mathbf{E}(|S_{E_i}^{(r-1)}|)\nonumber\\
&= \sum_{i\in V} \mathbf{E}(|S_i^{(r-1)}|) f p_i + f \mathbf{E}(|S_E^{(r-1)}|)\label{expalledges}\\
%\end{align}
%
%Likewise, summing Equation~\ref{expnode} over all nodes, by linearity of expectation,
%\begin{align}
\mathbf{E}(|S_V^{(r)}|) 
%&= \sum_{i\in V} \mathbf{E}(|S_i^{(r)}|) \nonumber\\
%&= \sum_{i\in V} \mathbf{E}(|S_i^{(r-1)}|) ( 1- f p_i ) + \sum_{i\in V} (1-f) \sum_{j\in N_i} \mathbf{E}(|S_{ji}^{(r-1)}|)\nonumber\\
&= \sum_{i\in V} \mathbf{E}(|S_i^{(r-1)}|) ( 1- f p_i ) + (1-f) \mathbf{E}(|S_E^{(r-1)}|).\label{expallnodes}
\end{align}

%Putting~\ref{expalledges} and~\ref{expallnodes} together,
%\begin{align*}
%\mathbf{E}(|S_E^{(r)}|)
%&= \sum_{i\in V} \mathbf{E}(|S_i^{(r-1)}|) f p_i + f \mathbf{E}(|S_E^{(r-1)}|)\\
%\mathbf{E}(|S_V^{(r)}|) 
%&= \sum_{i\in V} \mathbf{E}(|S_i^{(r-1)}|) ( 1- f p_i ) + (1-f) \mathbf{E}(|S_E^{(r-1)}|).
%\end{align*}

Using that $p_i\leq 1/2$ in~\ref{expalledges} and~\ref{expallnodes}, we have
\begin{align*}
\mathbf{E}(|S_E^{(r)}|)
&\leq (f/2) \mathbf{E}(|S_V^{(r-1)}|) + f \mathbf{E}(|S_E^{(r-1)}|)\\
\mathbf{E}(|S_V^{(r)}|) 
&\geq ( 1- f/2 ) \mathbf{E}(|S_V^{(r-1)}|) + (1-f) \mathbf{E}(|S_E^{(r-1)}|).
\end{align*}

Then,
\begin{align*}
\frac{\mathbf{E}(|S_E^{(r)}|)}{\mathbf{E}(|S_V^{(r)}|)}
&\leq \frac{(f/2) \mathbf{E}(|S_V^{(r-1)}|) + f \mathbf{E}(|S_E^{(r-1)}|)}{( 1- f/2 ) \mathbf{E}(|S_V^{(r-1)}|) + (1-f) \mathbf{E}(|S_E^{(r-1)}|)}\\
&\leq \frac{f}{1-f}, \textrm{ because } \frac{1-f/2}{1-f} \geq \frac{1}{2}.
\end{align*}

Then, given that $\mathbf{E}(|S_V^{(r)}|) + \mathbf{E}(|S_E^{(r)}|) = |S|$, we have $\mathbf{E}(|S_V^{(r)}|) \geq (1-f) |S|$.
%\begin{align*}
%\mathbf{E}(|S_V^{(r)}|) &\geq \frac{1-f}{f} \mathbf{E}(|S_E^{(r)}|)\\
%%\mathbf{E}(|S_V^{(r)}|) \left(\frac{1}{f}-\frac{1-f}{f}\right)&\geq \frac{1-f}{f} \mathbf{E}(|S_E^{(r)}|)\\
%\mathbf{E}(|S_V^{(r)}|) &\geq (1-f) |S|.
%\end{align*}
As proved in Lemma~\ref{lemma:rw}, 
with probability at least $1-1/n$, for any round $r\geq r_c/(1-q)$, $\max_{x\in S,i\in V}|p_x(i)-1/n|\leq \xi/n$, where $p_x(i)$ is the probability that particle $x$ is located at node $i$ and $q$ as defined in such lemma.
Then, for any node $i\in V$, it is $(1-\xi)(1-f) |S|/n \leq \mathbf{E}(|S_i^{(r)}|) \leq (1+\xi)|S|/n$ and the claim follows.

\end{proof}

%%%%%%%%%%%%%%%%%%%%%%%%%%%%%%%%%%%%%%%%%%%%%%%%%%%

Based on the previous lemmata, the following theorem shows the convergence time of \MDFU.
\begin{theorem}
\label{thm:fail}
Consider any communication network of $n$ nodes running \MDFU. 
For any $0 < f \leq 1/\ln(2\Delta e)^3$, let $q=1/e$ if $f\leq 1/(e(2\Delta e)^e)$, or $q=f \left(\sqrt{4\ln(2\Delta e)^3/f-3}-1\right)/2$ otherwise,
and let $r_c = 2 \ln(n/\xi)/\Phi(G)^2$.
%\begin{align*}
%q = \left\{ \begin{array}{ll} 
%1/e & \textrm{ if $f\leq 1/(e(2\Delta e)^e)$}\\
%f \left(\sqrt{4\ln(2\Delta e)^3/f-3}-1\right)/2  & \textrm{ otherwise.}
%\end{array} \right.
%\end{align*}
Then, with probability at least $1-1/n$, 
for any $0<\xi<1$ and any round $r\geq r_c/(1-q)$, 
the expected 
%relative error at any node $i\in V$ is $|\mathbf{E}(e_i^{(r)})-\overline{v}|/\overline{v}\leq 1-(1-f)(1-\xi)$,
average estimation at any node $i\in V$ is $(1-\xi)(1-f) \overline{v} \leq \mathbf{E}(e_i^{(r)}) \leq (1+\xi)\overline{v}$,
where $\Phi(G)$ is the conductance of the underlying graph characterizing the execution of \MDFU on the network.
\end{theorem}

\begin{proof}
From Lemmas~\ref{lemma:rw} and~\ref{lemma:error}, we know that, 
under the conditions of this theorem,
for any round $r\geq r_c/(1-q)$ and any node $i\in V$, 
with probability at least $1-1/n$
the expected number of particles (of the multiple random walk modeling \MDFU) is $(1-\xi)(1-f) |S|/n \leq \mathbf{E}(|S_i^{(r)}|) \leq (1+\xi)|S|/n$. 
Then, multiplying by the value of each particle the claim follows.
%, we obtain $(1-\xi)(1-f) \overline{v} \leq \mathbf{E}(e_i^{(r)}) \leq (1+\xi)\overline{v}$. 
%\begin{align*}
%(1-\xi)(1-f) |S|/n \leq \mathbf{E}(|S_i^{(r)}|) \leq (1+\xi)|S|/n.\\
%(1-\xi)(1-f) \overline{v} \leq \mathbf{E}(e_i^{(r)}) \leq (1+\xi)\overline{v}.
%\end{align*}
\end{proof}

\section{Empirical Evaluation of MDFU}
\label{sec:eval}
We evalutated \MDFU in a synchronous network simulator, using an
Erdős–Rényi\cite{erdos:randomgraphs} network with 1000 nodes and 5000 links
(giving an average degree of 10).
The input values were chosen as when performing node counting~\cite{JMBagg};
i.e., all values being 0 except a random node with value 1; this scenario is
more demanding, leading to slower convergence, than uniformly random input
values.
The evaluation aimed at: 1) comparing its convergence speed under no loss with
competing algorithms; 2) evaluating its behavior under message loss; 3)
checking its ability to perform continuous estimation over time-varying
input values.% (described in the appendix).

\subsection{Convergence Speed Against Related Algorithms Under no Faults}

To evaluate wether \MDFU is a practical algorithm in terms of convergence
speed, we compared it against three other algorithms: the original \FUname~\cite{ABJ:flowupdate,ABJ:flowupdatedyn}(FU), Distributed Random
Grouping~\cite{CPX:SNaggJournal} (DRG), and
Push-Synopses~\cite{KDGgossip}. Figure~\ref{fig:mdfu-others} shows the coefficient of
variation of the root mean square error as a function of the number of
rounds (averaging 30 runs), with CV(RMSE) = $\sqrt{\sum_{i \in V} (e_i -
\overline{v})^2 / n} / \overline{v}$.
%\footnote{CV(RMSE) = $\sqrt{(\sum_{i \in V} (e_i - \overline{v})^2) / N} /
%\overline{v}$.}

It can be seen that \MDFU is competitive, providing approximate estimates
slightly faster than FU and DRG and giving reasonably accurate results
roughly in line with them. It loses to them for very high precision
estimation and to Push-Synopses for all precisions (but both DRG and
Push-Synopses are not fault-tolerant).

\begin{figure}[t]
\centering
\includegraphics[width=0.49\textwidth]{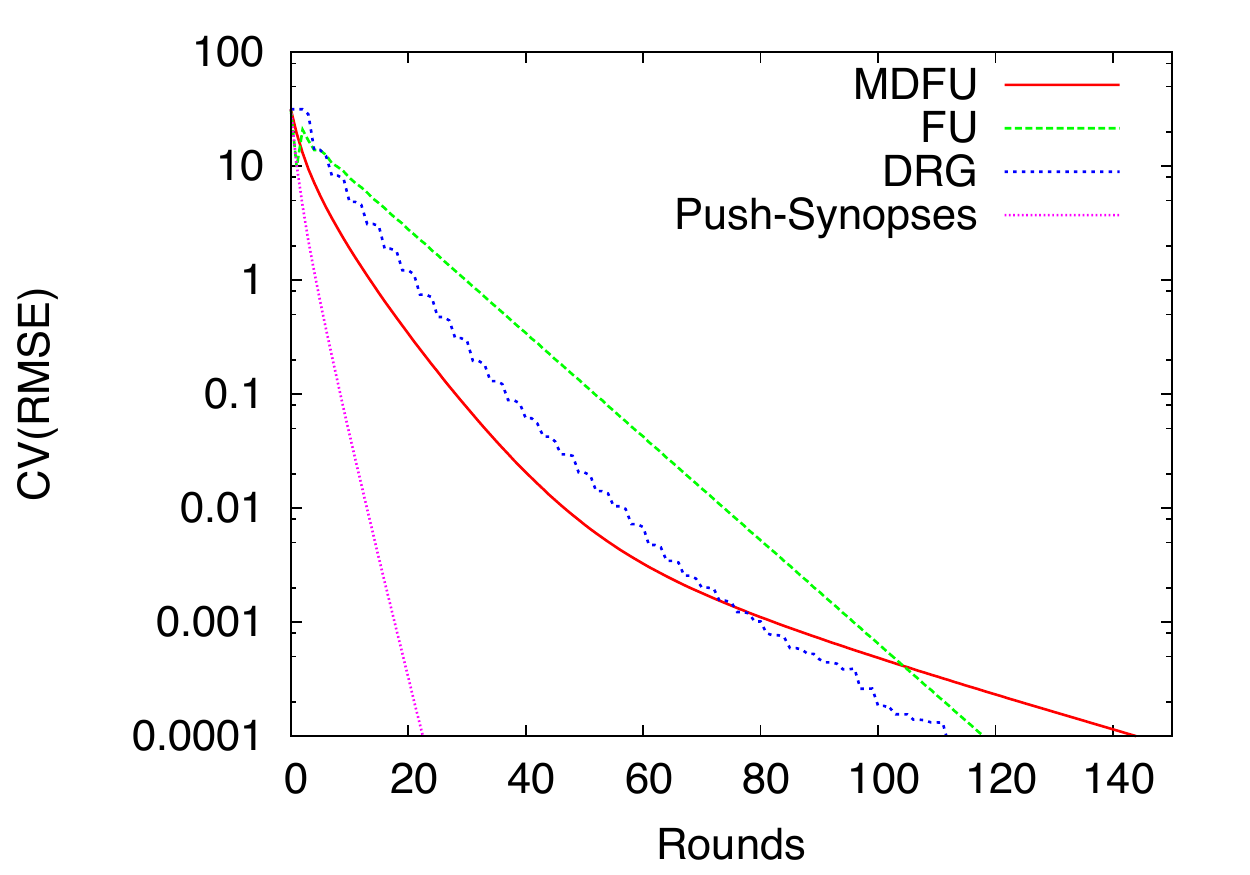}
\caption{CV(RMSE) over rounds in a 1000 node 5000 link Erdos-Renyi network.}
\label{fig:mdfu-others}
\end{figure}

\subsection{Fault Tolerance: Resilience to Message Loss}

To evaluate the resilience of \MDFU to message loss, we performed
simulations using different rates of message loss (0, 1\%, 5\%, 10\%), where
each individual message may fail to reach the destination with these given
probabilities. We measured the effect of message loss on both the CV(RMSE)
and also on the maximum relative error. As can be seen in
Figure~\ref{fig:cvrmse-max-loss}, as long as there is some message loss,
they do not tend to zero anymore, but converge to a value that is a
function of the message loss rate.

\begin{figure}[t]
	\centering
	\includegraphics[width=0.49\textwidth]{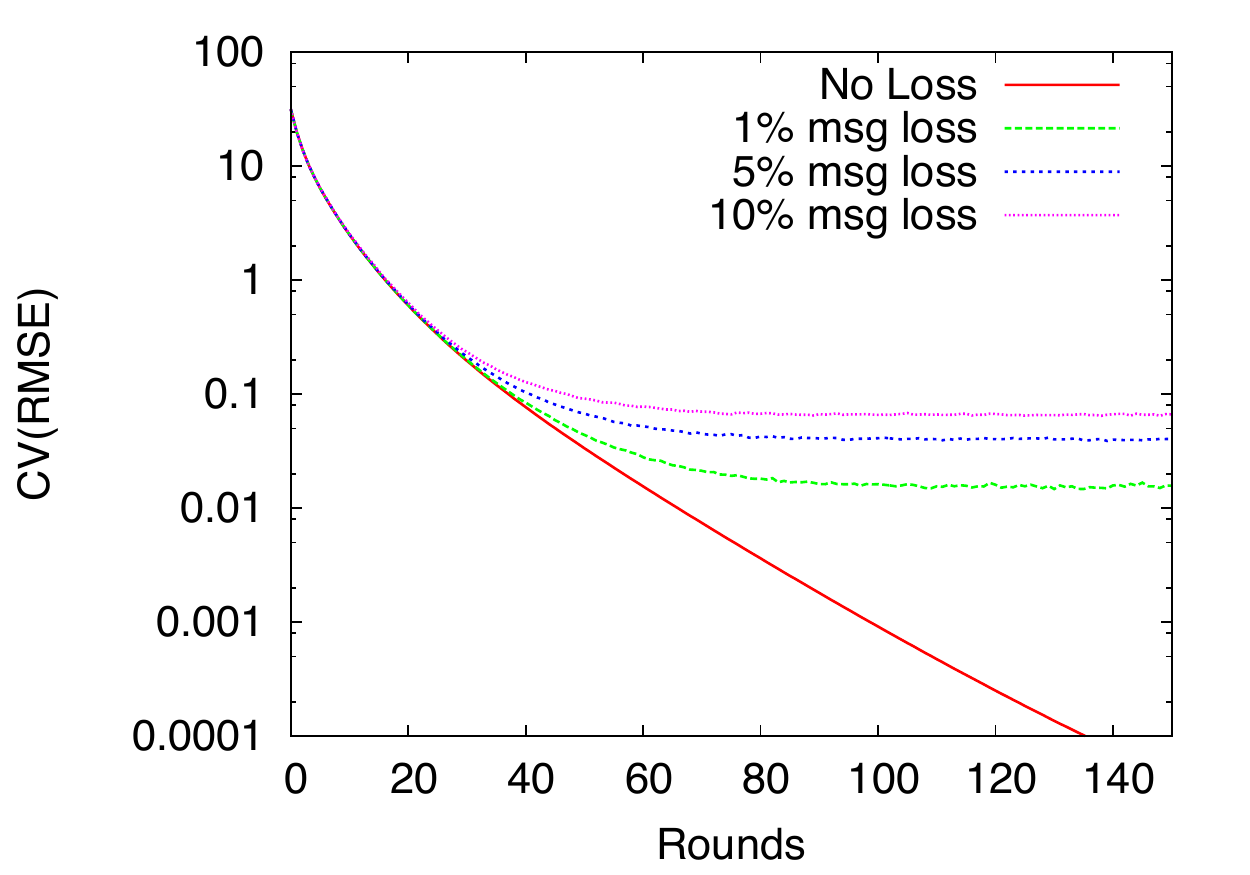}
	\includegraphics[width=0.49\textwidth]{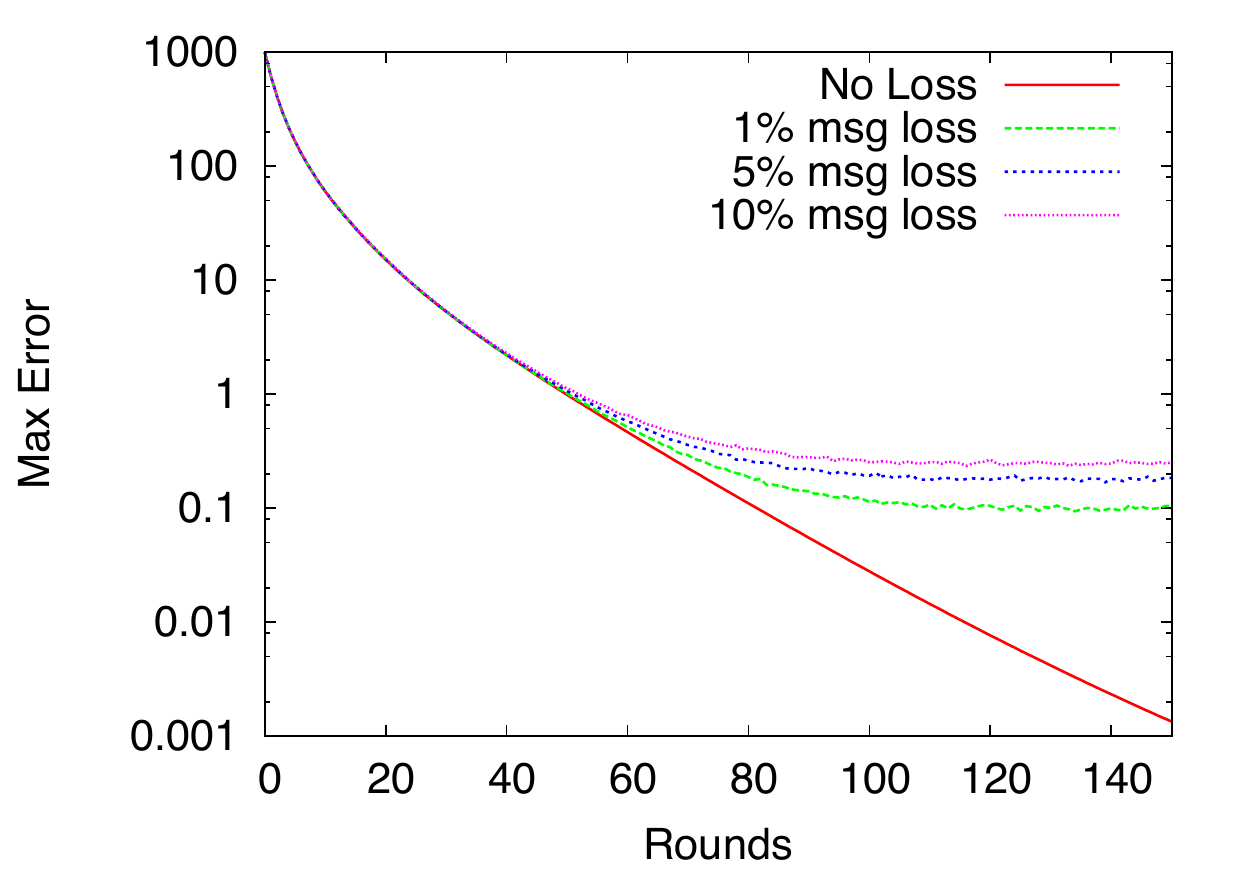}
	\caption{Coefficient of variation of the RMSE and maximum relative
	error for \MDFU in a 1000 node 5000 link Erdos-Renyi network.}
	\label{fig:cvrmse-max-loss}
\end{figure}

We also measured the behavior of the average of the estimates over the whole
network, and observed that there is a deviation from the correct value
($\overline v$, the average of the input values) towards lower values.
Figure~\ref{fig:mdfu-bias} shows the relative deviation from the correct
value over time, for different message loss rates.
It can be seen that this bias is roughly proportional to the message
loss rate (for these small message loss rates).

Relating these results with the theoretical analysis of \MDFU, we can see
that this bias should not come as a surprise. From Theorem~\ref{thm:fail},
the expected value of the estimation converges to a band between $(1 -
f)\overline v$ and $\overline v$. The relative deviation of
the lower boundary is thus proportinal to the message loss rate.
Figure~\ref{fig:mdfu-bias} also shows this boundary for the different
message loss rates.

This kind of bias was not present in the original FU, in which the average
of the estimates tends to the correct value. In \MDFU the message loss rate
limits the precision that can be achieved, but it does not impact
convergence, contrary to classic mass distribution algorithms where,
given message loss, the more rounds pass, the more mass is lost and the more
the estimates deviate from the correct value, failing to converge.

%We will show in the next section that this problem can be overcome, allowing
%arbitrary precision even for very large message loss rates.

\begin{figure}[t]
\centering
\includegraphics[width=0.49\textwidth]{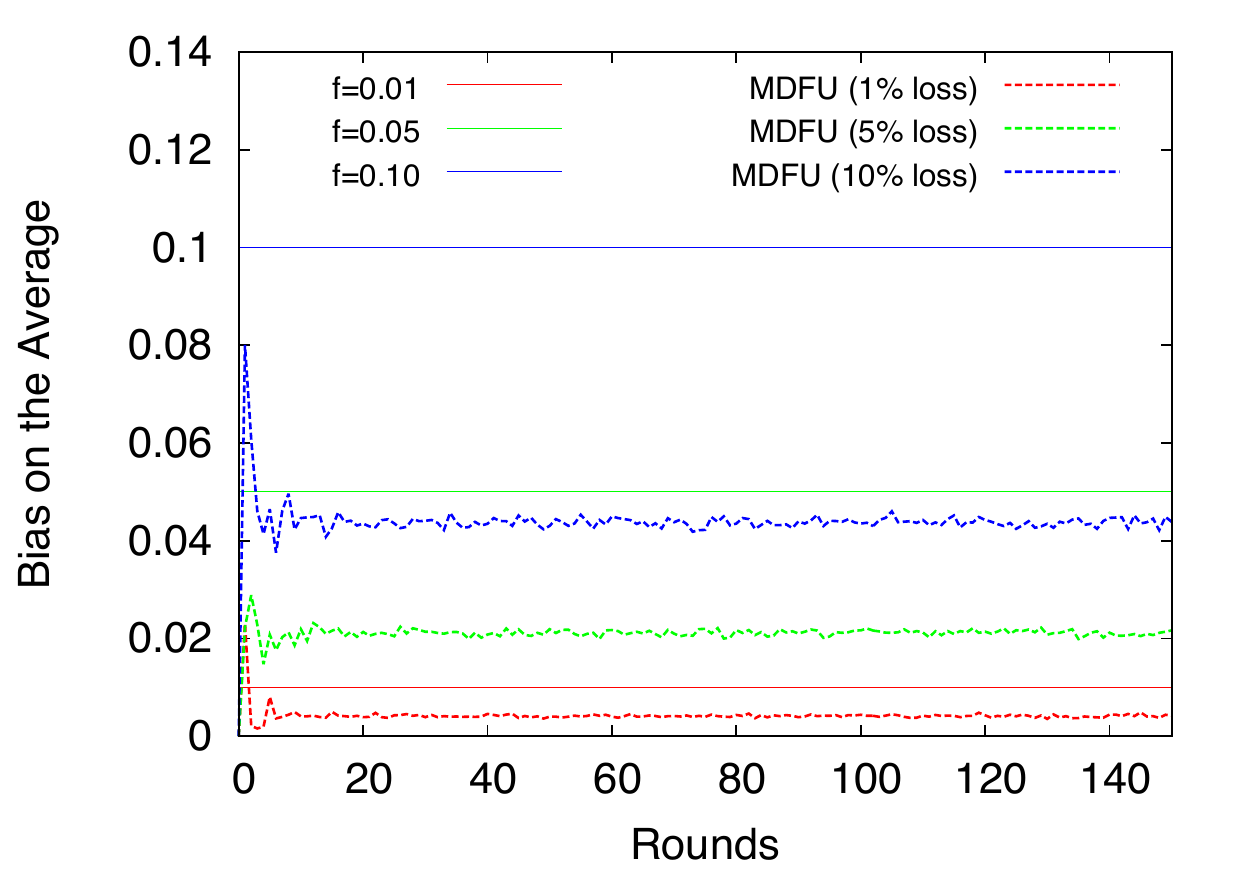}
\caption{Bias on the average estimation over rounds in a 1000 node 5000 link
	Erdos-Renyi network.}
\label{fig:mdfu-bias}
\end{figure}

\section{\MDFU with Linear Prediction}
\label{sec:vel}
The explanation for the behavior of \MDFU under message loss lies in that
only the estimate converges, but flows keep steadily increasing over time.
This can be seen in the formula:
%\[
$F_{out}(j) \leftarrow F_{out}(j) + e_i/\left(2D_{ij}\right)$
%\]
where the flow sent to some neighbor increases at each round by a value
depending on the estimate and their mutual degrees. What happens is that
during convergence, the extra flow that each of two nodes send over a link
tend to the same value, and the extra outgoing flow cancels out the extra
incoming flow. We can say that it is the \emph{velocity} (rate of increase)
of flows over a link that converge (to some different value for each link).

This means that, even if the estimate had already converged to the correct
value, given a message loss, the extra flow that should have been received
is not added to the estimate, implying a discrete deviation from the correct
value. This discrete deviation does not converge to zero; thus, we have a
bias towards lower values and the relative estimation error is prevented
from converging to zero given some message loss rate.

Here we improve \MDFU by exploring \emph{velocity convergence}. We keep,
for each link, the velocity (rate of increase) of the flow received. If a
message is lost, we predict what would have been the flow received, given
the stored flow, the velocity and the rounds passed since the last message
received over that link, i.e., we perform a \emph{linear prediction} of
incoming flow. When a message is received we update the flow and recalculate
the velocity. This algorithm is presented in Algorithm~\ref{mdfulp-alg}.

\begin{algorithm}[t]%[htbp]
\label{mdfulp-alg}
\caption{\MDFULP. Pseudocode for node $i$. $e_i$ is the estimate of node
$i$. $F_{in}(j)$ is the cumulative inflow from node $j$. $F_{out}(j)$ is the
cumulative outflow to node $j$. $V(j)$ is the velocity of incoming flow from
node $j$. $R(j)$ is the number of rounds since the last message received
from node $j$.
}
\dontprintsemicolon

\tcp{initialization}
$e_i \leftarrow v_i$\;
\ForEach{$j\in N_i$}{
$F_{in}(j) \leftarrow 0$\;
$F_{out}(j) \leftarrow e_i/\left(2D_{ij}\right)$\;
$V(j) \leftarrow 0$\;
$R(j) \leftarrow 1$\;
}
\BlankLine\;

\ForEach{round}{
\tcp{communication phase}
\ForEach{$j\in N_i$}{
Send $j$ message $\langle i,F_{out}(j)\rangle$\;
}
\tcp{computation phase}
\ForEach{$\langle j,F\rangle$ received}{
$V(j) \leftarrow (F - F_{in}(j)) / R(j)$\;
$R(j) \leftarrow 0$\;
$F_{in}(j) \leftarrow F$\;
}
$e_i \leftarrow v_i + \sum_{j\in N_i} (F_{in}(j) + V(j) \times R(j) - F_{out}(j)) $\;
\ForEach{$j\in N_i$}{
$F_{out}(j) \leftarrow F_{out}(j) + e_i/\left(2D_{ij}\right)$\;
$R(j) \leftarrow R(j) + 1$\;
}
}
\end{algorithm}

Under no message loss \MDFULP is the same as \MDFU and the theoretical
results on convergence speed also apply to \MDFULP. Under message loss the
velocities converge over time and the prediction will be increasingly more
accurate. Therefore, message loss should not cause discrete deviations in
the estimate, allowing the estimation error to converge to zero.

We have evaluated \MDFULP for the same network as before, but now with a
wide range of message loss rates. We have observed that the behavior under
message loss rates below 50\% is almost indistinguishable from the behavior
under no message loss.
Figure~\ref{fig:cvrmse-loss-lp} shows the CVRMSE and maximum relative error
for 0\%, 60\%, 70\%, and 80\% message loss rates. It can be seen that even for
60\% loss rate, after 60 rounds we have basically the same estimation errors
as under no message loss.

\begin{figure}[t]
	\centering
	\includegraphics[width=0.49\textwidth]{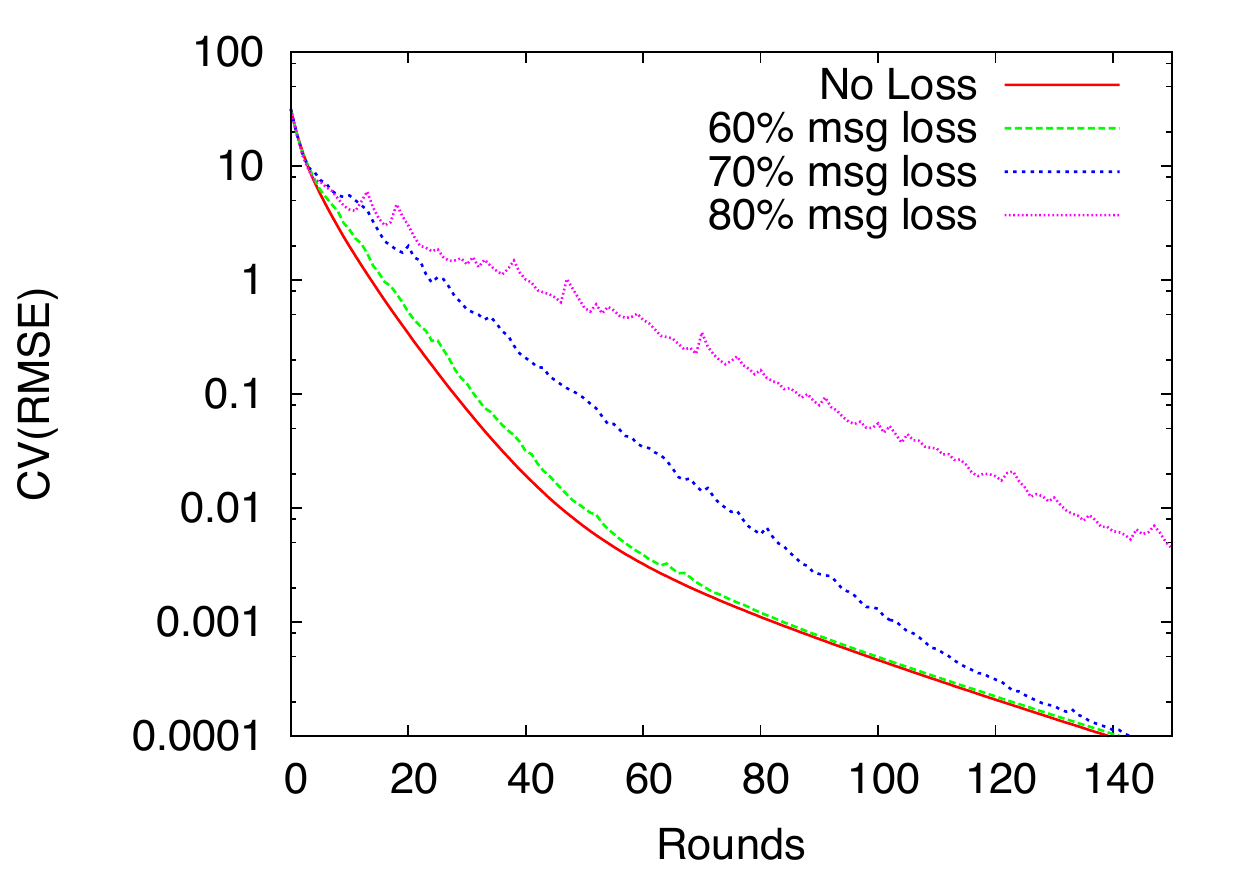}
	\includegraphics[width=0.49\textwidth]{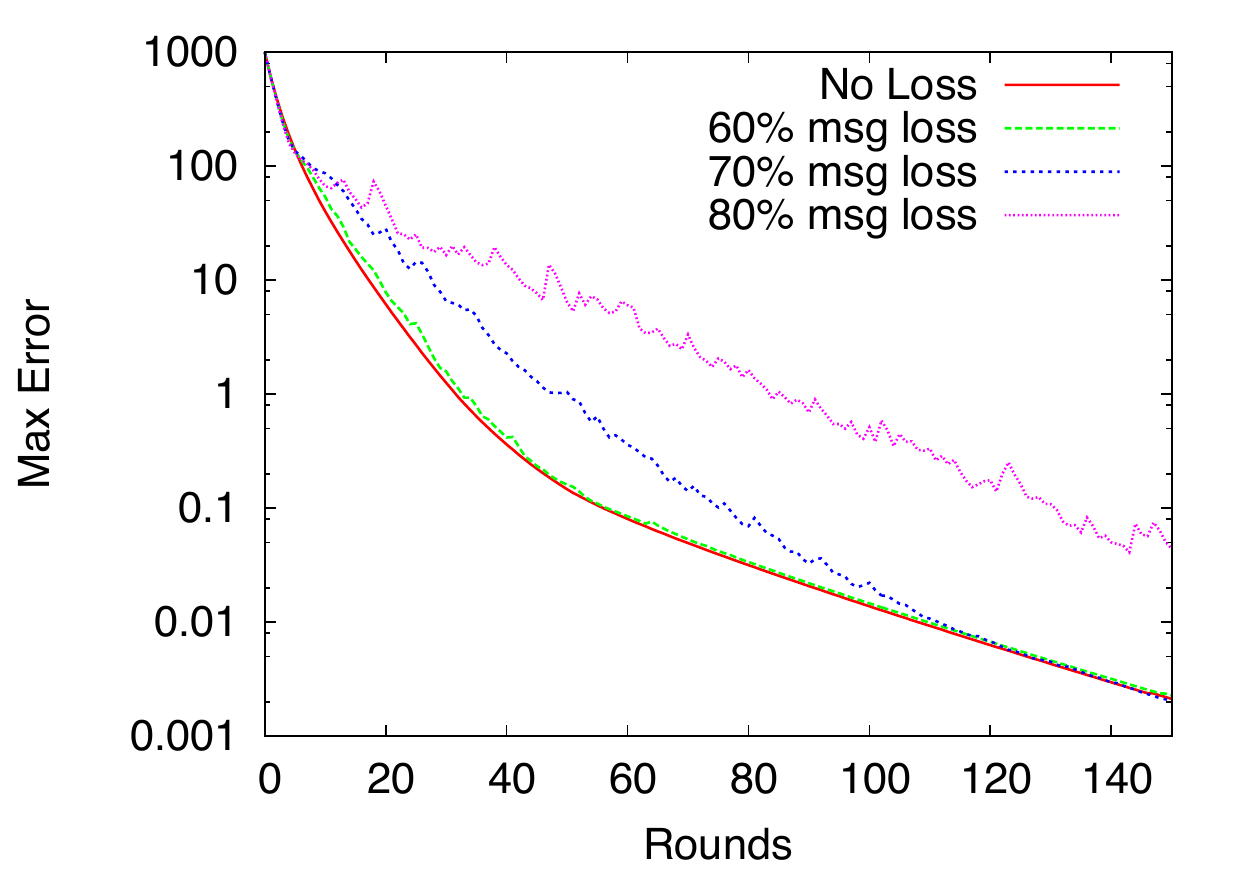}
	\caption{Coefficient of variation of the RMSE and maximum relative
	error for \MDFULP in a 1000 node 5000 link Erdos-Renyi network.}
	\label{fig:cvrmse-loss-lp}
\end{figure}

%\section*{Appendix}

\section{Continuous Estimation Over Time-Varying Input Values}
\label{sec:value-change}

\begin{figure}[t]
\centering
\includegraphics[width=0.49\textwidth]{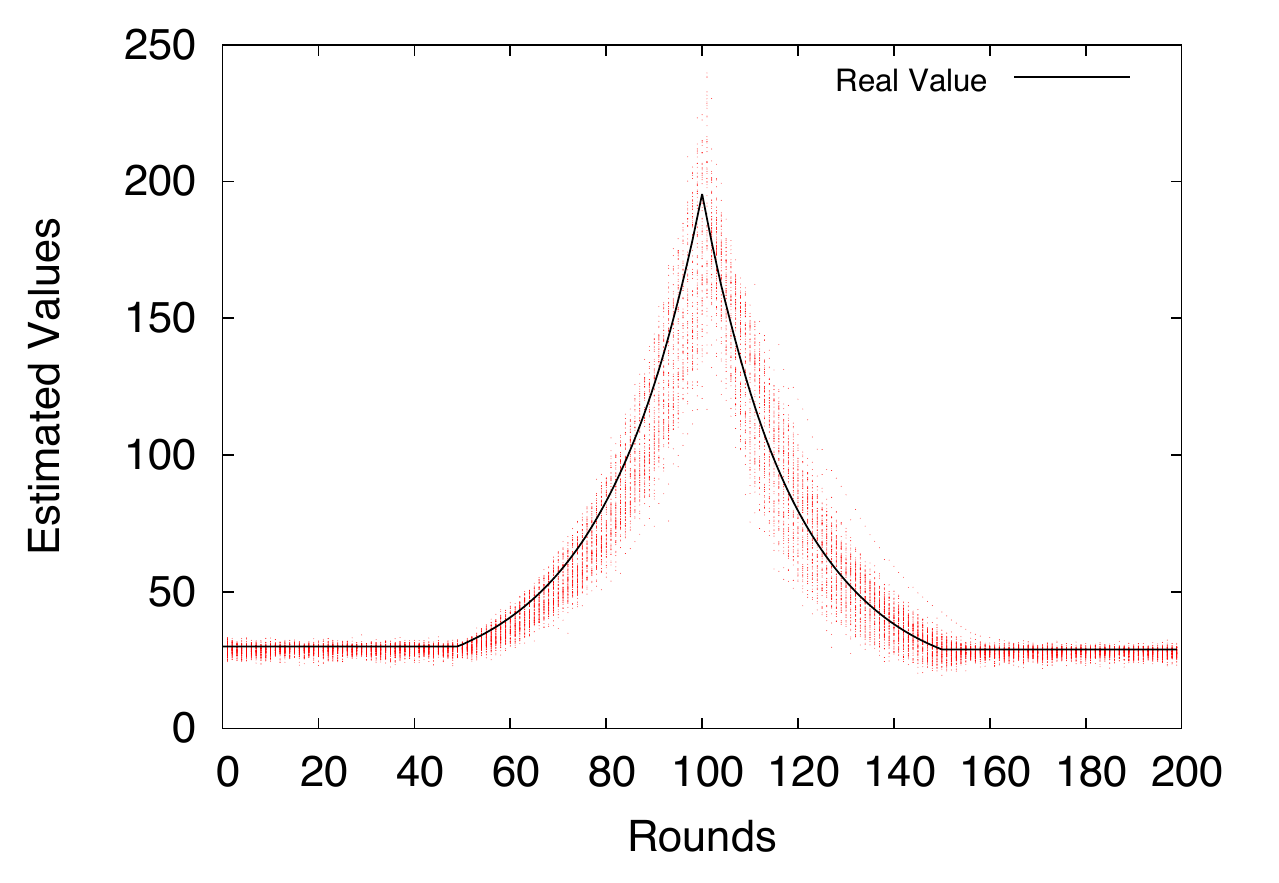}
	\caption{Estimated value over rounds in a 1000 node 5000 link
	Erdos-Renyi network, with changes of the initial input value at 50\% of the nodes.}
	\label{fig:mdfu-loss-cvrmse}
\end{figure}

Up to thus point we have considered that the input values $v_i$ are fixed
throughout the computation. In most practical situations this will not be
the case and input values will change along time. The common approach in MD
algorithms is to periodically reset the algorithm and start a new run that
freezes the new input values and aggregates the new average. Naturally,
resets are inefficient and mechanisms that can adapt the ongoing computation
have the potential to adjust the estimates in a much shorter number of
rounds.

Without any further modifications, MDFU (and MDFU-LP) share with FU the
capability of adapting to input value changes, since $v_i$ is considered in
the computation of the local estimate $e_i$, and this
regulates how much the outgoing flows are to be incremented. If $v_i$
decreases, $e_i$ decreases in the same proportion and node $i$ will share
less through its flows to the neighbours. The converse occurring when $v_i$
increases. The overall effect is convergence to the new average, even if
multiple nodes are having changes in their input values. 

In Figure \ref{fig:mdfu-loss-cvrmse} we show an example of how MDFU handles
input value changes. In this setting, starting at round 50 and during 50
rounds, we increase by 5\% in each round the input value in 500 nodes (a
random half of the 1000 nodes). In the following 50 rounds, the same 500
nodes will have its value decreased by 5\% per round. Initial input values
are chosen uniformly at random (from 25 to 35) and the run is made with
message loss at 10\%. In Figure  \ref{fig:mdfu-loss-cvrmse} one can observe
that individual estimates\footnote{To avoid clutering the graph only shows
individual estimate evolution for a random sample of 100 of the 1000 nodes.}
closely follow the global average, with only a slight lag of some rounds. 

Notice that the lag could never be zero, since we are updating the new
global average (black line) instantaneously and even the fastest theoretical
algorithm would need information that takes \emph{diameter} rounds to acquire.

\begin{small}
\bibliographystyle{plain}
\bibliography{./Comprehensive}

\begin{thebibliography}{10}

\bibitem{akyildiz-survey}
I.~F. Akyildiz, W.~Su, Y.~Sankarasubramaniam, and E.~Cyirci.
\newblock Wireless sensor networks: A survey.
\newblock {\em Computer Networks}, 38(4):393--422, 2002.

\bibitem{BGMGM:agg}
M.~Bawa, H.~Garcia-Molina, A.~Gionis, and R.~Motwani.
\newblock Estimating aggregates on a peer-to-peer network.
\newblock Technical report, Stanford University, Database group, 2003.

\bibitem{BGPSgossip}
Stephen Boyd, Arpita Ghosh, Balaji Prabhakar, and Devavrat Shah.
\newblock Randomized gossip algorithms.
\newblock {\em IEEE/ACM Transactions on Networking}, 14(SI):2508--2530, 2006.

\bibitem{CH:aggSHS}
Jen-Yeu Chen and Jianghai Hu.
\newblock Analysis of distributed random grouping for aggregate computation on
  wireless sensor networks with randomly changing graphs.
\newblock {\em IEEE Trans. Parallel Distr. Syst.}, 19(8):1136--1149, 2008.

\bibitem{Chen09}
Jen-Yeu Chen, Gopal Pandurangan, and Jianghai Hu.
\newblock Brief announcement: locality-based aggregate computation in wireless
  sensor networks.
\newblock In {\em PODC '09: Proceedings of the 28th ACM symposium on Principles
  of distributed computing}, pages 298--299, New York, NY, USA, 2009. ACM.

\bibitem{CPX:SNaggJournal}
Jen-Yeu Chen, Gopal Pandurangan, and Dongyan Xu.
\newblock Robust computation of aggregates in wireless sensor networks:
  distributed randomized algorithms and analysis.
\newblock {\em IEEE Trans. Parallel Distr. Syst.}, 17(9):987--1000, 2006.

\bibitem{CPMS:dynRN}
A.E.F. Clementi, F.~Pasquale, A.~Monti, and R.~Silvestri.
\newblock Communication in dynamic radio networks.
\newblock In {\em Proc. 26th Ann. {ACM} Symp. on Principles of Distributed
  Computing}, pages 205--214, 2007.

\bibitem{DSWgeoGossJournal}
A.~G. Dimakis, A.D. Sarwate, and M.J. Wainwright.
\newblock Geographic gossip : Efficient averaging for sensor networks.
\newblock {\em IEEE Transactions on Signal Processing}, 56(3):1205--1216, 2008.

\bibitem{erdos:randomgraphs}
P.~Erdos and A.~Renyi.
\newblock On random graphs--i.
\newblock {\em Publicationes Matematicae}, 6:290--297, 1959.

\bibitem{FMT:agg}
A.~{Fern\'andez~Anta}, M.~A. Mosteiro, and C.~Thraves.
\newblock An early-stopping protocol for computing aggregate functions in
  sensor networks.
\newblock In {\em Proc. of the IEEE 15th Pacific Rim International Symposium on
  Dependable Computing}, pages 357--364, 2009.

\bibitem{G:gossiping}
Leszek Gasieniec.
\newblock Randomized gossiping in radio networks.
\newblock In Ming-Yang Kao, editor, {\em Encyclopedia of Algorithms}. Springer,
  2008.

\bibitem{DBLP:conf/dsn/GuptaRB01}
Indranil Gupta, Robbert van Renesse, and Kenneth~P. Birman.
\newblock Scalable fault-tolerant aggregation in large process groups.
\newblock In {\em DSN}, pages 433--442. IEEE Computer Society, 2001.

\bibitem{DBLP:conf/sosp/HeidemannSIGEG01}
John~S. Heidemann, Fabio Silva, Chalermek Intanagonwiwat, Ramesh Govindan,
  Deborah Estrin, and Deepak Ganesan.
\newblock Building efficient wireless sensor networks with low-level naming.
\newblock In {\em SOSP}, pages 146--159, 2001.

\bibitem{intanagonwiwat00directedJournal}
C.~Intanagonwiwat, R.~Govindan, D.~Estrin, J.~Heidemann, and F.~Silva.
\newblock Directed diffusion for wireless sensor networking.
\newblock {\em IEEE/ACM Transactions on Networking}, 11(1):2--16, 2003.

\bibitem{DBLP:conf/icdcs/IntanagonwiwatEGH02}
Chalermek Intanagonwiwat, Deborah Estrin, Ramesh Govindan, and John~S.
  Heidemann.
\newblock Impact of network density on data aggregation in wireless sensor
  networks.
\newblock In {\em ICDCS}, pages 457--458, 2002.

\bibitem{JMBagg}
M\'{a}rk Jelasity, Alberto Montresor, and Ozalp Babaoglu.
\newblock Gossip-based aggregation in large dynamic networks.
\newblock {\em ACM Transactions on Computer Systems}, 23(3):219--252, 2005.

\bibitem{ABJ:flowupdatedyn}
P.~Jesus, C.~Baquero, and P.S. Almeida.
\newblock Fault-tolerant aggregation for dynamic networks.
\newblock In {\em Proc. of the 29th IEEE Symposium on Reliable Distributed
  Systems}, pages 37--43, 2010.

\bibitem{ABJ:flowupdate}
Paulo Jesus, Carlos Baquero, and Paulo Almeida.
\newblock Fault-tolerant aggregation by flow updating.
\newblock In {\em Proc. of the 9th IFIP WG 6.1 International Conference
  Distributed Applications and Interoperable Systems}, volume 5523 of {\em
  Lecture Notes in Computer Science}, pages 73--86. Springer, 2009.

\bibitem{KDGgossip}
D.~Kempe, A.~Dobra, and J.~Gehrke.
\newblock Gossip-based computation of aggregate information.
\newblock In {\em Proc. of the 44th IEEE Ann. Symp. on Foundations of Computer
  Science}, pages 482--491, 2003.

\bibitem{KBCHLrobust}
G.~Kollios, J.~W. Byers, J.~Considine, M.~Hadjieleftheriou, and F.~Li.
\newblock Robust aggregation in sensor networks.
\newblock {\em IEEE Data Engineering Bulletin}, 28(1):26--32, 2005.

\bibitem{KP:adaptvsobliv}
D.~R. Kowalski and A.~Pelc.
\newblock Time complexity of radio broadcasting: adaptiveness vs. obliviousness
  and randomization vs. determinism.
\newblock {\em Theoretical Computer Science}, 333:355--371, 2005.

\bibitem{DBLP:conf/icdcsw/KrishnamachariEW02}
Bhaskar Krishnamachari, Deborah Estrin, and Stephen~B. Wicker.
\newblock The impact of data aggregation in wireless sensor networks.
\newblock In {\em ICDCS Workshops}, pages 575--578. IEEE Computer Society,
  2002.

\bibitem{MFHHtag}
Samuel Madden, Michael~J. Franklin, Joseph~M. Hellerstein, and Wei Hong.
\newblock Tag: a tiny aggregation service for ad-hoc sensor networks.
\newblock In {\em Proc. of the 5th Symp. on Operating Systems Design and
  Implementation}, pages 131--146, 2002.

\bibitem{MSFCagg}
Samuel Madden, Robert Szewczyk, Michael~J. Franklin, and David Culler.
\newblock Supporting aggregate queries over ad-hoc wireless sensor networks.
\newblock In {\em Proceedings of the Fourth IEEE Workshop on Mobile Computing
  Systems and Applications}, page~49, 2002.

\bibitem{book:mitzenmacher}
M.~Mitzenmacher and E.~Upfal.
\newblock {\em Probability and Computing: Randomized Algorithms and
  Probabilistic Analysis}.
\newblock Cambridge University Press, 2005.

\bibitem{NGSAagg}
Suman Nath, Phillip~B. Gibbons, Srinivasan Seshan, and Zachary~R. Anderson.
\newblock Synopsis diffusion for robust aggregation in sensor networks.
\newblock In {\em Proceedings of the 2nd international conference on Embedded
  networked sensor systems}, pages 250--262, 2004.

\bibitem{OSM:ConsProb}
Reza Olfati-Saber and Richard~M. Murray.
\newblock Consensus problems in networks of agents with switching topology and
  time-delays.
\newblock {\em Transactions on Automatic Control}, 49(9):1520--1533, 2004.

\bibitem{ravi-survey}
P.~Rentala, R.~Musumuri, U.~Saxena, and S.~Gandham.
\newblock Survey on sensor networks.
\newblock http://citeseer.nj.nec.com/479874.html.

\bibitem{SP:DistAvgCons}
Dzulkifli~S. Scherber and Haralabos~C. Papadopoulos.
\newblock Locally constructed algorithms for distributed computations in ad-hoc
  networks.
\newblock In {\em Proceedings of the 3rd International Symposium on Information
  Processing in Sensor Networks}, pages 11--19, 2004.

\bibitem{SJ:count}
Alistair Sinclair and Mark Jerrum.
\newblock Approximate counting, uniform generation and rapidly mixing markov
  chains.
\newblock {\em Information and Computation}, 82(1):93--133, 1989.

\bibitem{Spanos:2005ux}
D~Spanos, R~Olfati-Saber, and R~Murray.
\newblock {Dynamic consensus on mobile networks}.
\newblock In {\em 16th IFAC World Congress}, 2005.

\bibitem{XB:distAvgCons}
Lin Xiao and Stephen Boyd.
\newblock Fast linear iterations for distributed average.
\newblock {\em Systems and Control Letters}, 53:65--78, 2004.

\bibitem{XBL:distSensorFusion}
Lin Xiao, Stephen Boyd, and Sanjay Lall.
\newblock A scheme for robust distributed sensor fusion based on average
  consensus.
\newblock In {\em Proceedings of the 4th International Symposium on Information
  Processing in Sensor Networks}, pages 63--70, 2005.

\bibitem{Xiao:2006jn}
Lin Xiao, Stephen Boyd, and Sanjay Lall.
\newblock {A Space-Time Diffusion Scheme for Peer-to-Peer Least-Squares
  Estimation}.
\newblock In {\em Proceedings of the 5th International Conference on
  Information Processing in Sensor Networks}, pages 168--176, 2006.

\bibitem{ZGEagg}
Jerry Zhao, Ramesh Govindan, and Deborah Estrin.
\newblock Computing aggregates for monitoring wireless sensor networks.
\newblock In {\em Proc. of the 1st IEEE Intl. Workshop on Sensor Network
  Protocols and Applications}, pages 139--148, 2003.

\end{thebibliography}
\end{small}

\end{document}